\documentclass[11pt,a4paper]{article}

\usepackage[margin=1in]{geometry}
\usepackage[utf8]{inputenc}
\usepackage{microtype}
\usepackage{amsmath,amsthm,amsfonts,amssymb}
\usepackage{xcolor}
\usepackage{graphicx}
\usepackage{algorithm}
\usepackage[numbers,sort&compress]{natbib}
\usepackage{url}
\usepackage[colorlinks=true,linkcolor=blue,citecolor=blue,urlcolor=blue,hypertexnames=false]{hyperref}

\theoremstyle{plain}
\newtheorem{theorem}{Theorem}

\newtheorem{claim}[theorem]{Claim}
\newtheorem{observation}[theorem]{Observation}

\theoremstyle{definition}
\newtheorem{definition}[theorem]{Definition}
\theoremstyle{remark}
\newtheorem{remark}[theorem]{Remark}

\DeclareMathOperator{\dist}{dist}
\DeclareMathOperator{\arb}{arb}
\DeclareMathOperator{\poly}{poly}

\newcommand{\algname}[2]{\hyperref[#1]{\textcolor{black}{\textsf{#2}}}}
\newcommand{\localmis}{\algname{alg.MIS}{Local-Simulation-Greedy-MIS}}
\newcommand{\tester}{\algname{alg:tolerant}{Tolerant-Test-Subgraph-freeness}}
\newcommand{\approxMIS}{\algname{alg:approx}{Approximate-Average-MIS-In-Auxiliary-Graph}}
\newcommand{\approxKk}{\algname{alg:approxKk}{Approximate-$K_k$-copies-in-$G_\theta$}}
\newcommand{\randomKk}{\algname{alg:randomKk}{Get-Random-$K_k$-Copy-in-$G_\theta$}}
\newcommand{\allNKk}{\algname{alg:allNKk}{$K_k$-All-Neighbors-in-$G_\theta$}}
\newcommand{\approxH}{\algname{alg:approxH}{Approximate-$H$-copies-in-$G_\theta$}}
\newcommand{\randomH}{\algname{alg:randomH}{Get-Random-$H$-Copy-in-$G_\theta$}}
\newcommand{\allNH}{\algname{alg:allnH}{$H$-All-Neighbors-in-$G_\theta$}}

\title{Tolerant Testers for Subgraph-Freeness\thanks{A preliminary version of this paper appeared in the proceedings of the 33rd Annual European Symposium on Algorithms (ESA 2025).}}

\author{
  Reut Levi\thanks{Efi Arazi School of Computer Science, Reichman University, Israel. Email: \texttt{reut.levi1@idc.ac.il}. Supported by the Israel Science Foundation under Grant 1867/20.}
  \and
  Jonathan Meiri\thanks{Efi Arazi School of Computer Science, Reichman University, Israel. Email: \texttt{jonathan.meiri@post.runi.ac.il}. Supported by the Israel Science Foundation under Grant 1867/20.}
}
\date{}

\begin{document}
\maketitle

\begin{abstract}
    In this paper we study the problem of tolerantly testing the property of being $H$-free (which also implies distance approximation from being $H$-free).
    
 In the general-graphs model, we show that tolerant $K_k$-freeness testing can be achieved with query complexity that is polynomial in the arboricity of the input graph $G$, $\arb(G)$, and independent of the size of $G$ (for graphs in which the average degree is $\Omega(1))$. 

        Specifically for triangles, our algorithm distinguishes graphs which are $\epsilon$-close to being triangle-free from graphs that are $3\epsilon(1+\eta)$-far from being triangle-free with expected query complexity which is $\tilde{O}(\arb^3(G))$ (for constant $\eta$ and $\epsilon$).
        
    For general $k$-cliques our algorithm distinguishes graphs which are $\epsilon$-close to being $K_k$-free from graphs which are ${k \choose 2}\epsilon(1+\eta)$-far from being $K_k$-free with expected query complexity which is polynomial in $k$, $1/\epsilon$, $1/\eta$ and $\arb(G)$. 

    We then generalize our result and provide a similar result for any motif $H$ which is $2$-connected and has radius $1$. This includes for example the wheel graph. 

    Finally, we show that our tester can be applied to the bounded-degree model for tolerantly testing $H$-freeness for any motif $H$. The query complexity of the algorithm is polynomial in the degree bound, $d$, improving the previous state-of-the-art by Marko and Ron (TALG 2009) that obtained quasi-polynomial query complexity in $d$.

\end{abstract}

\noindent\textbf{Keywords:} tolerant testing, property testing, subgraph freeness, distance approximation, arboricity.

\section{Introduction}
    The problem of testing $H$-freeness, where $H$ is a small {\em motif} (namely, of size which is independent of the input graph), has received a lot of attention due to its fundamental importance for various fields including biology, sociology and network science. Detecting specific patterns, such as subgraphs, can reveal community structures, vulnerabilities, or anomalous behavior.
    
    In the realm of property testing the decision problem is relaxed such that we are only aiming to distinguish graphs which are $H$-free from graphs that are {\em far} from being $H$-free (according to some predetermined distance measure).  
    However, handling real-world data often means dealing with noise, incompleteness, or small errors. In these scenarios, tolerant testing (introduced by Parnas, Ron and Rubinfeld~\cite{PARNAS20061012}) for $H$-freeness becomes crucial, as it allows us to distinguish graphs that are nearly $H$-free from those that are far from being $H$-free, thereby bypassing the problem of imperfections in the data. More specifically, a tolerant property testing algorithm is required to distinguish objects that are 
    $\epsilon_1$-close to having a given property $\mathcal{P}$
     from objects that are 
    $\epsilon_2$-far from having 
    $\mathcal{P}$, for some parameters $0 \leq \epsilon_1 < \epsilon_2 \leq 1$. Clearly, by definition, tolerant testing is at least as hard as testing. In fact, the separation between tolerant and non-tolerant testing can be very dramatic~\cite{FischerF06}. Another benefit of tolerant testing is that it implies (by a simple reduction) distance approximation where the ratio of approximation depends on the parameters $\epsilon_1$, $\epsilon_2$ (see more details in~\cite{PARNAS20061012}). Therefore, whenever feasible, we should aim to develop tolerant testers, ideally without incurring a significant overhead in query complexity.  
    
One of the early works on tolerant testing by Marko and Ron~\cite{MR09} showed that it is possible to tolerantly test $H$-freeness in the bounded-degree model~\cite{GR-bound} for any motif $H$ with query complexity that is quasi-polynomial in the degree bound and independent of the size of the graph.

In the regime of non-tolerant testers, Levi~\cite{reut} showed that it is possible to test triangle-freeness with query complexity that depends linearly on the arboricity  of the graph (which is tight for graphs with arboricity $O(\sqrt{n})$) and is independent of the size of the graph (assuming the average degree is $\Omega(1)$). Since graphs of bounded degree have bounded arboricity, this result applies to a broader family of graphs. However, this tester is not tolerant.

One natural question is whether the result in~\cite{reut} can be extended to show that tolerant testing of triangle-freeness (and ideally other motifs) is possible in the general-graphs model with query complexity that depends only on the arboricity of the graph. Such a result would also generalize~\cite{MR09} (since bounded-degree graphs inherently have bounded arboricity).
Another important question is whether the query complexity of subgraph freeness in the bounded-degree model can be improved to be polynomial in the degree bound (rather than quasi-polynomial). In this paper, we answer both questions affirmatively.

\subsection{Our Results}
\subsubsection{The general-graphs model}
\paragraph{Tolerant Testing $K_k$-freeness} Our tester for the property of being $K_k$-free  distinguishes graphs that are $\epsilon$-close to being $K_k$-free from graphs that are $\epsilon {k\choose 2}(1+\eta)$-far from being $K_k$-free, where $\epsilon\in (0,1]$ and $\eta\in (0,1]$ are parameters.
    The expected query complexity of the algorithm is $\tilde{O}\left(\frac{k^4\theta^{2k-3}}{\eta^2\epsilon \bar{d}} + k^5\theta^{2k-4}\theta^{\min\{2, k-2\}}\right)$ where $\theta = \Theta(\arb(G)/(\eta\epsilon))$ and $\bar{d}$ is the average degree (see Theorem~\ref{thm:Kk}).

Specifically, assuming $\bar{d} = \Omega(1)$, for the case of  triangles, our algorithm distinguishes graphs which are $\epsilon$-close to being triangle-free from graphs that are $3\epsilon(1+\eta)$-far from being triangle-free with expected query complexity which is $\tilde{O}(\arb^3(G))\cdot \poly(\epsilon^{-1}, \eta^{-1},k)$.
 
For $k$-cliques where $k> 3$, our algorithm distinguishes graphs which are $\epsilon$-close to being $K_k$-free from graphs which are ${k \choose 2}\epsilon(1+\eta)$-far from being $K_k$-free with expected query complexity which is $\tilde{O}(\arb^{2k-2}(G))\cdot \poly(\epsilon^{-1}, \eta^{-1},k)$. 

\paragraph{Tolerant $H$-freeness for $H$ which is $2$-connected with radius $1$}
We extend our result and apply our tester for motifs that are $2$-connected with radius $1$. This includes for example the wheel graph\footnote{Any connected graph with an additional vertex that is incident to all other vertices is $2$-connected with radius $1$.}.
Specifically our tester distinguishes graphs that are $\epsilon$-close to being $H$-free from graphs that are $\epsilon |E(H)|(1+\eta)$-far from being $H$-free.
    The expected query complexity of the algorithm is $\tilde{O}\left(\left(h^7/\eta^2+ h^3/(\eta \epsilon \bar{d})\right)\cdot \theta^{3h-2}\right)$ where $\theta = \Theta(\arb(G)/(\eta\epsilon))$ and $h = |V(H)|$ (see Theorem~\ref{thm:H}).

    Specifically, assuming $\bar{d} = \Omega(1)$, the expected query complexity of our tester is $\tilde{O}(\arb^{3h-2}(G))\cdot \poly(\epsilon^{-1}, \eta^{-1},h)$.

\subsubsection{The bounded-degree graphs model}
\paragraph{Tolerant $H$-freeness for any $H$}
    We show that our tolerant tester can be applied to the bounded-degree model to test $H$-freeness for any motif $H$. Specifically the tester distinguishes graphs that are $\epsilon$-close to being $H$-free from graphs that are $\epsilon |E(H)|(1+\eta)$-far from being $H$-free.
    The expected query complexity of the algorithm is $\tilde{O}\left(\left(h^7/\eta^2+ h^3/(\eta \epsilon \bar{d})\right) \cdot d^{3h-2}\right)$ where $h = |V(H)|$ and $d$ is the degree bound.

    Specifically, assuming $\bar{d} = \Omega(1)$, the expected query complexity of our tester is $\tilde{O}(d^{3h-2})\cdot \poly(\epsilon^{-1}, \eta^{-1},h)$.

    This improves the result in~\cite{MR09} that has query complexity that is quasi-polynomial in $d$ to polynomial dependence in $d$.

    \begin{remark}
        As noted in~\cite{MR09}, the result of Bogdanov, Obata and Trevisan~\cite{BogdanovOT02} implies a linear lower bound on the query complexity of approximating the distance from being triangle-free for some small multiplicative and/or
additive error (even in bounded degree graphs). Therefore, we cannot hope to obtain tolerant testers of subgraph freeness for all parameters $\epsilon_1, \epsilon_2$.
    \end{remark}

\begin{remark}
    As shown in~\cite{EdenLR24}, testing $C_4$-freeness requires $\Omega(n^{1/4})$ queries. Therefore we cannot hope to obtain tolerant testers (in the general-graphs model) for general motifs whose query complexity depends only on the arboricity of the graph. 
\end{remark}

\subsection{Our Algorithm}
In this section we describe our algorithm for tolerantly testing $K_k$-freeness. This gives a good perspective on our approach. The first ingredient of our algorithm, which also appears in~\cite{reut}, is to perform the test on a  subgraph of $G$. This subgraph, which we refer to as $G_{\theta}$ is the graph after removing the edges between $\theta$-heavy vertices, where a vertex is $\theta$-heavy if its degree is greater than $\theta$ where $\theta = \Theta(\arb(G)/(\eta \epsilon))$. By doing this modification we might decrease the distance by at most $\Theta(\eta \epsilon)$, so if the graph is sufficiently far from being $K_k$-free, we will still have enough witnesses of violation. On the positive side, since this subgraph does not have edges in which both endpoints are $\theta$-heavy, any $K_k$ in this graph will include at most a single $\theta$-heavy vertex and therefore can be revealed by exploring the neighborhood of the non-heavy vertices of the clique. 

The second ingredient of our algorithm is to approximate the number of copies of $K_k$ in $G_\theta$ (it will become clear later why we need to approximate this parameter). This can be done by sampling vertices, $v$, u.a.r. from $G$, checking if the degree of $v$ is bounded by $\theta$ and if so selecting $k-1$ random indices in $[\theta]$. We then perform $k-1$ neighbor-queries on $v$, according to the selected indices, and check if the subgraph induced on these neighbors (if such neighbors exist) induces a $(k-1)$-clique. Thus, our sample space is of size $\Theta(n\theta^{k-1})$ and if the graph is $\epsilon$-far from being $K_k$-free then we have at least $\epsilon m$ copies of $K_k$ so we need $\Theta(n\theta^{k-1}/(\epsilon m))$ (which is $O(\arb(G)^{k-1}/\epsilon)$ when $\bar{d} = \Omega(1)$) trials to get a good approximation.  

Before we describe the last ingredient of our algorithm we shall define the {\em $H$-copies graph} of $G_{\theta}$, which we denote by $(G_\theta)^H$. In this graph the vertex set is  the set of $H$-copies in $G_\theta$ and a pair of copies are adjacent iff they share at least one edge. Consider a maximal independent set (MIS), $S$, of this graph. By definition, the copies of this set are edge-disjoint. Thus the minimum number of edges we need to remove from $G_{\theta}$ in order to make it $H$-free, which we denote by $\dist_H(G_\theta)$, is at least $|S|$. On the other hand, by the maximality of $S$ it follows that $\dist_H(G_\theta) \leq |E(H)||S|$ (since we can remove all $H$-copies in $G_\theta$ by removing all the edges of the copies in $S$). Therefore if we can get a good approximation of {\em any} MIS of $(G_\theta)^H$, then we will get a good approximation (up to a $|E(H)|$-factor) to the distance of $G_\theta$ from being $H$-free. This is also true for any MIS that is constructed  greedily according to some order on the vertices. 

Fortunately, as shown by the ingenious work of Yoshida, Yamamoto and Ito~\cite{YYI12}, if we pick a random vertex $v$ in a graph and a random ordering on the vertices of the graph then the expected number of recursive calls we need to perform in order to locally simulate the outcome of the greedy MIS algorithm on that graph (according to the random ordering we picked) equals the average degree of the graph plus $1$. 
Therefore, we can get a good approximation to the {\em average} probability that a random $H$-copy belongs to a random greedy MIS (which is determined by a random ordering) by executing the local simulation on random $H$-copies and random orderings ({\em both} selected uniformly and independently on each trial). 
Since the ratio between any MIS in a graph $F$ to the total number of vertices in $F$ is at most $1/(\Delta(F)+1)$, where $\Delta(F)$ denotes the maximum degree in $F$, we obtain that it suffices to evaluate the outcome of a  greedy MIS (each time according to a new random ordering) on $\Theta(\Delta((G_\theta)^H))$ copies selected uniformly at random. This will give us a good approximation to the {\em average} probability that a random $H$-copy belongs to a random greedy MIS. When we multiply this by the approximation to the total number of $H$-copies we previously obtained, this will provide a good approximation to the distance of $G$ from being $H$-free.  

\subsubsection{The extension for other motifs}
We describe the above mentioned algorithm as a general framework for tolerant testing of $H$-freeness and then implement (according to the defined interface) the necessary specific procedures for testing $K_k$-freeness and  then provide more general procedures for $H$-freeness of any motif $H$ which is $2$-connected and has radius $1$. Finally, we show that the procedures for the latter family of motifs can be used to test $H$-freeness for any motif $H$ in the bounded degree model.   
\subsection{Related Work}

\subsubsection{Tolerant testing of $H$-freeness in the bounded degree model}
As mentioned-above Marko and Ron~\cite{MR09} studied the problem of tolerantly testing $H$-freeness in the bounded degree model. They begin by describing a global algorithm that constructs a cover of edges by first performing $O(\log d)$ iterations in which the edges of some of the $H$-copies are added to the cover by a random process (that can be viewed as a distributed randomized algorithm for constructing a large independent set of copies) and then a single edge is added to the cover from each one of the uncovered copies. Since the global algorithm has a distributed nature they show that they can locally simulate its outcome by performing $d^{O(\log d)}$ queries and therefore obtain a good approximation to the size of the cover which gives a good approximation to the distance from being $H$-free.

\subsubsection{Non-tolerant testing of subgraph-freeness}
In the bounded degree model Goldreich and Ron~\cite{GR-bound} observed that it is possible to test triangle freeness with query complexity $O(1/\epsilon)$ in graphs of maximum degree bounded by some constant. 

The problem of  testing triangle freeness in the general graph model was first studied by Alon, Kaufman, Krivelevich and Ron~\cite{AKKR08}.
The query complexity of their algorithms depends on $n$ and $\bar{d}$, the number of vertices in the graph and the average degree, respectively.  
They provided sublinear upper bounds for almost the entire range of parameters. Moreover, their upper bounds are at most quadratic in their lower bounds. However, they are tight only when $d_{\rm max} = O(\bar{d})$ and $\bar{d} \leq \sqrt{n}$, where $d_{\rm max}$ denotes the maximum degree and $\bar{d}$ denotes the average degree of the graph or when $\bar{d} = \Theta(1)$.
Shortly after, Rast~\cite{Ras06} and Gugelmann~\cite{Gug06} 
 improved the upper bound and lower bound of~\cite{AKKR08}, respectively, for some ranges of the parameters.

More recently, Levi~\cite{reut} studied the complexity of testing triangle freeness as a function of the arboricity of the graph (in the general graphs model). It was shown that the complexity of testing triangle freeness depends linearly on the arboricity of the graph for graphs in which the arboricity is $O(\sqrt{n})$ and the average degree is $\Omega(1)$.

Even more recently, Eden, Levi and Ron~\cite{EdenLR24} studied the problem of testing $C_k$-freeness and showed that if the motif is not a $k$-clique then it is not always possible to provide a tester whose query complexity is independent of the size of the graph.
In particular, they showed that testing $C_4$-freeness requires $\Omega(n^{1/4})$ queries.

\subsubsection{Sublinear algorithms that receive the arboricity of the graph as a parameter}
Several sublinear-time graph algorithms for counting and sampling 
give improved results when the graph $G$ has bounded arboricity. This includes the following results.
Eden, Ron and Rosenbaum~\cite{ERR19} designed an algorithm for sampling edges almost uniformly.
 
Eden, Ron and Seshadhri~\cite{ERS19} estimate the degree distribution moments of an undirected graph and in particular estimate the average degree of a graph. 

In another paper, Eden, Ron and Seshadhri~\cite{ERS20} give a $(1\pm \epsilon)$-approximation for the number of $k$-cliques.
In a more recent work, Eden, Ron, and Rosenbaum~\cite{ERR22} showed that in this setting sampling cliques is harder than counting.

\section{Preliminaries}\label{sec:prelim}
We consider the general-graphs model~\cite{PR,KKR04}, where the algorithm can perform the following types of queries:
% is provided with the
%number $n$ of vertices in $G$ and 
\begin{enumerate}
\item For any  $v\in V$, query for $v$'s degree;
\item For any vertex $v\in V$ and index $i$, %\in [d(v)]$, 
query for $v$'s $i$-th neighbor (if $i> d(v)$, then a special symbol is returned);
\item For any pair of vertices $u,v$, query whether $\{u,v\}\in E$.
\end{enumerate}

We denote by $\Delta(G)$ the maximum degree of the graph $G$.
We denote by $\delta_H(G)$ the distance of $G$ from being $H$-free, namely, the fraction of edges that we need to remove from $G$ to make it $H$-free. We denote by $\dist_H(G) = \delta_H(G) |E(G)|$, namely, the minimum number of edges we need to remove from $G$ in order to make it $H$-free. For a graph $F$ we use $V(F)$ to denote its vertex set and $E(F)$ to denote its edge set. 

When we say with high constant probability (w.h.c.p.) we mean we can adjust the high constant probability without changing the asymptotic complexity of the algorithm.

\begin{definition}[Copies of a motif $H$]\label{def:copy}
 For a graph $G$ and a motif $H$,  we say that a subgraph $\mathcal{H}$ of $G$ is a \textsf{copy} of $H$ in $G$ if $\mathcal{H}$ is isomorphic to $H$.
\end{definition}

We use $n_H(G)$ to denote the number of $H$-copies in $G$.

Given a threshold $\theta$, we call the vertices, $v$, in $G$ such that $d_G(v) > \theta$, {\em $\theta$-heavy} and otherwise we call them {\em $\theta$-light}. 

\begin{definition}[\bf The graph $G_{\theta}$]
For a graph $G = (V, E)$ and a threshold $\theta$, the graph $G_{\theta}$ is the graph whose vertex set is $V$ and its edge set is: $$E\setminus \{\{u, v\} : u \text{ and } v \text{ are $\theta$-heavy}\}.$$
\end{definition}

\begin{definition}
    The arboricity of a graph $G$, denoted by $\arb(G)$, is the minimum number of forests into which the edges of the graph can be partitioned.
\end{definition}
It is known that for $\theta = \arb(G)/\eta$ where $\eta\in (0,1/2]$,  $|E(G_\theta)| \geq (1-2\eta)m$ (see, e.g., Claim 4 in~\cite{reut}).

We assume that the arboricity of the input graph $G$ as well as the number of edges, $m$, and vertices, $n$, are known to the algorithm. 

\begin{remark}
If the algorithm is not provided with the arboricity of the graph, then it may run the procedure from~\cite{reut} that obtains the \emph{effective arboricity} of the graph. 
More specifically, it was shown in~\cite{reut} how to obtain a value $\alpha^*$ that with high constant probability satisfies the following: (1) $\alpha^* \leq 2\arb(G)$; (2) The number of edges between vertices whose degree is at least $\Theta(\alpha^*/\epsilon)$ is at most $\epsilon m$, which suffices for our needs.
    Up to polylogarithmic factors in $n$ and polynomial factors in $1/\epsilon$, the query complexity of the procedure is $O(\arb(G))$ assuming the average degree is $\Omega(1)$. 
\end{remark}

\subsection{The $H$-copies graph of $G$}
In order to describe our algorithm we shall use the following definitions.
\begin{definition}\label{def:gh}
    Given a graph $G=(V, E)$ and a motif $H$ we define the $H$-copies graph of $G$, $G^H$ to be the graph whose vertex set are the copies of $H$ in $G$, $\mathcal{H}_1,\ldots, \mathcal{H}_\ell$ and $e = \{\mathcal{H}_1, \mathcal{H}_2\}$ belongs to the edge set of $G^H$ iff $\mathcal{H}_1$  and $\mathcal{H}_2$ share at least one edge.
\end{definition}

Next, we define the notion of an estimator for the number of $H$-copies of a graph $G$. The estimator is given a guessed lower bound on $n_H(G_\theta)$. If the guess is accurate, the estimator produces a reliable estimate; otherwise, it does not significantly overshoot, as detailed in the following definition.

\begin{definition}\label{def:nHEst}
    We say that an algorithm is a {\em $H$-copies number estimator} for a graph $G$ if it receives as parameters $s$, a guess of a lower bound on $n_H(G)$, and $\gamma$, the approximation parameter, and possibly other parameters for which the following holds:
    If $s \leq n_{H}(G)$ then w.h.c.p. the return value of the algorithm is in $(1\pm \gamma) n_{H}(G)$. Otherwise, w.h.c.p. it is at most $(1+\gamma)s$. 
\end{definition}

\begin{remark}
    The role of the guess is  to provide some (indirect) control to the user of the estimator on the number of attempts to sample a copy from the graph by the  estimator (e.g. if the graph is $H$-free then the estimator might make too many attempts to sample an $H$-copy with no success).  
\end{remark}

\begin{definition}\label{def:copiesOracle}
An algorithm $\mathcal{A}$ is a {\em $H$-copies oracle} of a graph $G=(V,E)$ if it supports the following queries. \begin{enumerate}
    \item On query {\rm Random-copy} with parameters $t\in \mathbb{N}$ and $\delta\in (0,1]$, it returns a copy of $H$ in $G$ uniformly at random or fails. If the number of copies in $G$ is at least $t$ then it returns a copy w.p. at least $1-\delta$.
    
    \item On query {\rm All-neighbors}($\mathcal{H}$), where $\mathcal{H}$ is a copy of $H$ in $G$, it returns all the neighbors of  $\mathcal{H}$ in $G^H$ (see Definition~\ref{def:gh}). 
\end{enumerate}
\end{definition}

\subsection{Local Simulation of Greedy Maximal IS}
Our algorithm uses a local simulation of a greedy MIS algorithm. 
Algorithm~\ref{alg.MIS} describes this local simulation. For a vertex set $V$, we let $S_V$ denote the set of all permutations over $V$.

\begin{algorithm}\caption{\localmis }
\label{alg.MIS}
Input: Query access to a graph $G = (V, E)$ and parameters $\pi\in S_{V}$ and $v\in V$
\begin{enumerate}
    \item Perform an {\rm All-Neighbors} query on $v$ and sort its neighbors according to $\pi$. Let $v_1, \ldots, v_\ell$ denote its neighbors according to this order.  
    \item For $i=1,\ldots, \ell:$
    \begin{enumerate}
        \item If $v_i$ precedes $v$ in $\pi$, recursively call \localmis\ with parameters $\pi$ and $v_i$. If the return value is \emph{YES} then return \emph{NO}.
    \end{enumerate}
    \item Return \emph{YES}.
\end{enumerate}

\end{algorithm}

Let $R^G_{\pi}(v)$ denote the number of (recursive) calls to \localmis\ during the evaluation of \localmis\ on $G$, $\pi$ and $v$.
We shall use the following result from the work of Yoshida, Yamamoto and Ito~\cite{YYI12}.

\begin{theorem}[\cite{YYI12}] \label{YYI}
For any graph $G = (V, E)$ with $n$ vertices and $m$ edges, 
\begin{equation*}
\mathbb{E}_{\pi \in S_V, v \in V} [R^G_{\pi}(v)] \leq 1 + \frac{m}{n}\;.    
\end{equation*}
\end{theorem}

\section{The algorithm for tolerant testing of $H$-freeness}
Before describing our algorithm we first relate the distance from being $H$-free to the size of a Maximal-Independent-Set of $G^H$ as stated in the following claim. We then use this relation to provide a relation between the expectation of a random variable to the distance from being $H$-free.

\begin{claim}\label{clm:maxIS}
    For any maximal independent set of $G^H$, $S$, it holds that 
    \begin{equation*}
    |S| \leq \dist_H(G) \leq |S|\cdot |E(H)| .   
    \end{equation*}
\end{claim}
\begin{proof}
Since $S$ is an independent set in $G^H$, it follows that any pair $\mathcal{H}_1, \mathcal{H}_2 \in S$ is a pair of edge-disjoint copies of $H$ in $G$ (by definition). Since $S$ is maximal, it follows that if we remove for each $\mathcal{H} \in S$ all its edges from $G$ then $G$ becomes $H$-free (since we remove at least one edge from each copy of $H$ in $G$).
\end{proof}

\begin{claim}\label{clm:MIS}
    For $v$ that is drawn u.a.r. from $V(G^H)$ and $\pi$ which is drawn u.a.r. from $S_{V}$ it follows that: 
    \begin{equation}\label{eq:1}
        \Pr[v\in MIS_\pi(G^H)] \in \left[\frac{\dist_H(G)}{|E(H)| n_H(G)}, \frac{\dist_H(G)}{n_H(G)}\right].
    \end{equation}
    Moreover, 
    \begin{equation}\label{eq:Delta}
       \Pr[v\in MIS_\pi(G^H)] \geq 1/(\Delta(G^H)+1). 
    \end{equation}
\end{claim}
\begin{proof}
    By Claim~\ref{clm:maxIS}, for any maximal IS in $G^H$, $S$, it holds that $\frac{\dist_H(G)}{|E(H)|} \leq |S| \leq \dist_H(G)$.
    Therefore, for any such $S$ it holds that $\Pr[v\in S] \in \left[\frac{\dist_H(G)}{|E(H)| n_H(G)}, \frac{\dist_H(G)}{n_H(G)}\right]$, where $v$ is drawn u.a.r. from $V(G^H)$. This is true since $|V(G^H)| = n_H(G)$. Since $MIS_\pi(G^H)$ is maximal for any $\pi$ (by definition), the claim follows. 
    Equation~\eqref{eq:Delta} follows from the fact that each vertex $v$ in the independent set covers at most $d(v) + 1$ vertices from the graph. Thus, any maximal independent set is of size at least $n_H(G)/(\Delta(G^H)+1)$, in particular, $MIS_\pi(G^H)$ (for any $\pi$). This concludes the proof.
\end{proof}

By using the previous claims we arrive to the following claim. 
\begin{claim}\label{clm:approx}
    If $s\leq n_H(G)$ then with high constant probability, the return value of \approxMIS\ is in 
    \begin{equation}
        \left[(1-\gamma)\frac{\dist_H(G)}{|E(H)| n_H(G)}, (1+\gamma)\frac{\dist_H(G)}{n_H(G)}\right].
    \end{equation}
\end{claim}
\begin{proof}
Let $E_1$ denote the event that in all 
 the executions of Step~\ref{step:randomcopy} of the algorithm, a copy of $H$ was returned. By Definition~\ref{def:copiesOracle}, the setting of $\delta$ and since $s\leq n_H(G)$, w.h.c.p., $E_1$ occurs. We henceforth assume $E_1$ occurs. Consider the random variable $Z_i$ (see Step~\ref{step:zi} of the algorithm). By construction $Z_i$ takes values in $\{0,1\}$. By Claim~\ref{clm:MIS}, Equation~\eqref{eq:1}, $\mathbb{E}(Z_i)\in \left[\frac{\dist_H(G)}{|E(H)| n_H(G)}, \frac{\dist_H(G)}{n_H(G)}\right]$. By Equation~\eqref{eq:Delta}, $\mathbb{E}(Z_i) \geq 1/(\Delta(G^H)+1)$.   
    The claim follows from the setting of $\ell$ and the multiplicative Chernoff's bound (see Theorem~\ref{thm:chermul}).
\end{proof}

\begin{algorithm}\caption{\approxMIS}\label{alg:approx}
Input: $\gamma$ - approximation parameter, $s$ - guess on a lower bound on $n_H(G)$, $\Delta$ - an upper bound on $\Delta(G^H)$, and access to an $H$-copy oracle of $G$ (see Definition~\ref{def:copiesOracle})
\begin{enumerate}
\item Set $\ell = \Theta(\Delta/\gamma^2)$ and $\delta = \Theta(1/\ell)$.
\item             For $i = 1,\ldots,\ell$
        \begin{enumerate}
            \item Run an $H$-copy oracle of $G$ with parameters $s$ and $\delta$, and obtain an $H$-copy, $\mathcal{H}$, uniformly at random.\label{step:randomcopy}
            \item Draw $\pi$ u.a.r. and run Algorithm~\ref{alg.MIS} on $\mathcal{H}$, $\pi$ and the $H$-copies oracle of $G$ (see Definition~\ref{def:copiesOracle}). 
            \item Set $Z_i = 1$ if the algorithm returned YES and set $Z_i = 0$ otherwise. \label{step:zi}
        \end{enumerate}
        \item Return $\hat{Z} = \sum_{i\in [\ell]} Z_i/\ell.$
\end{enumerate}
\end{algorithm}

\begin{algorithm}
\caption{\tester}\label{alg:tolerant}
Input: $\epsilon$, $\eta$ and query access to a graph $G$
    \begin{enumerate}
%        \item Run the algorithm from Claim~\ref{lem:alpha} on $G$ and obtain $\arb^*$
        \item Set $\gamma = \eta/ 64$ and $\theta = \arb(G)/(\gamma\epsilon)$
        \item Run the number of $H$-copies estimator for the graph $G_\theta$ with parameters $\epsilon m/2$ and $\gamma$ (see Definition~\ref{def:nHEst}) and obtain $\hat{n}$. \label{step:est}
        \item If $\hat{n}$ is less than $\epsilon m$ then return ACCEPT.  \label{step:hat}
        \item Run \approxMIS\ on the graph $G_\theta$ with parameters $\epsilon m/2$ and $\gamma$ and $\Delta((G_\theta)^H)$. Let $\hat{Z}$ denote the returned value.
        \item If $\hat{n} \hat{Z} > \epsilon m (1+\gamma)^2$ then return REJECT. Otherwise, return ACCEPT.
    \end{enumerate}
\end{algorithm}

We next prove our main theorem for this section. In order to realize this theorem we provide in the next sections concrete $H$-copies oracles and $H$-copies number estimators (first for $H$ which is a $k$-clique and then generalize, for $H$ which is $2$-connected and has radius $1$).

\begin{theorem}\label{thm:general}
    Given parameters $\epsilon \in (0,1]$, $\eta\in (0,1]$ and query access to a graph $G$, \tester\ has the following guarantees. If $G$ is $\epsilon$-close to being $H$-free then w.h.c.p. it accepts $G$. If $G$ is $\epsilon |E(H)| (1+\eta)$-far from being $H$-free then w.h.c.p. it rejects $G$.
\end{theorem}
\begin{proof}
Assume $G$ is $\epsilon$-close to being $H$-free. 
If the algorithm accepts $G$ in Step~\ref{step:hat} then we are done.
Otherwise, it follows that $\hat{n} \geq \epsilon m$.

Let $E_1$ denote the event that the estimator in Step~\ref{step:est} returns a good approximation as described in Definition~\ref{def:nHEst}. Conditioned on $E_1$ and the fact that $\hat{n} \geq \epsilon m$ it is implied that $n_H(G_{\theta}) \geq \epsilon m /2$ (assume otherwise and reach a contradiction to the fact that $\hat{n} \geq \epsilon m$ as $(1+\gamma) < 2$) and hence its return value is in $(1\pm \gamma)n_H(G_{\theta})$. 
 
Let $E_2$ denote the event that \approxMIS\ returns a good approximation as described in Claim~\ref{clm:approx}.
Conditioned on $E_1\cap E_2$ we obtain that 
\begin{equation}\label{eq.nz}
\hat{n}\hat{Z} \in   \left[(1-\gamma)^2\frac{\dist_H(G_\theta)}{|E(H)|}, (1+\gamma)^2\dist_H(G_\theta)\right].
\end{equation}
Since $\dist_H(G_\theta) \leq \dist_H(G) \leq \epsilon m$ we obtain that $\hat{n}\hat{Z} \leq (1+\gamma)^2\epsilon m$ and hence the algorithm accepts $G$.

Assume $G$ is $\beta$-far from being $H$-free where $\beta = \epsilon |E(H)| (1+\eta)$. 
Since $\theta = \arb(G)/(\gamma\epsilon)$ and $\gamma\epsilon \leq 1/2$, at most $2\gamma\epsilon m$ edges of $G$ have both endpoints $\theta$-heavy (see Section~\ref{sec:prelim}), namely, $|E(G)\setminus E(G_\theta)| \leq 2\gamma\epsilon m$. Any set of edges whose removal makes $G_\theta$ $H$-free, together with $E(G)\setminus E(G_\theta)$, makes $G$ $H$-free, and hence $\dist_H(G) \leq \dist_H(G_\theta) + 2\gamma\epsilon m$. Since $\beta = \epsilon|E(H)|(1+\eta) \geq 2\epsilon$, we obtain $\dist_H(G_\theta) \geq \beta m - 2\gamma\epsilon m \geq (1-\gamma)\beta m$. Thus, $n_H(G_\theta) \geq (1-\gamma)\beta m$ (as every copy of $H$ in $G_\theta$ is destroyed by removing a single edge, the number of copies is at least the number of edges one must remove).
In particular, since $\beta \geq 2\epsilon$ and $\gamma \leq 1/2$, it holds that $n_H(G_\theta) \geq \epsilon m \geq \epsilon m/2$, i.e., the number of copies is at least the parameter $s = \epsilon m/2$ passed to the estimator in Step~\ref{step:est}. Therefore, conditioned on $E_1$ (defined as above), we obtain that $\hat{n} \in (1\pm \gamma) n_H(G_\theta)$, and in particular $\hat{n} \geq (1-\gamma)^2\beta m \geq \epsilon m$, so the algorithm does not accept in Step~\ref{step:hat}. Hence, conditioned on $E_1\cap E_2$, Equation~\eqref{eq.nz} holds in this case as well. In particular $\hat{n}\hat{Z}\geq (1-\gamma)^3 \frac{\beta m}{|E(H)|} > \epsilon m (1+\gamma)^2$ (since $(1+\eta) > (1+2\gamma)^5 > (1+\gamma)^2/(1-\gamma)^3$). Thus, conditioned on $E_1\cap E_2$, the algorithm rejects $G$, as desired. 
\end{proof}

\begin{claim}\label{clm:exp}
    The expected query complexity of \tester\ is bounded above by the query complexity of performing $\ell = \Theta(\Delta/\eta^2)$ random-copy queries to $(G_\theta)^H$ with parameters $s = \epsilon m/2$ and $\delta = \Theta(\eta^2/\Delta)$ (see Definition~\ref{def:copiesOracle}),  $\Theta(\ell \cdot \Delta)$ all-neighbor queries to $(G_\theta)^H$, where $\Delta = \Delta((G_\theta)^H)$, and a single execution of an $H$-copies number estimator for $G_\theta$ (see Definition~\ref{def:nHEst}). 
\end{claim}
\begin{proof}
    By construction, linearity of expectation and Theorem~\ref{YYI}.
\end{proof}

\section{The oracles for $k$-cliques}
In this section we describe our $K_k$-copies oracle and $K_k$-copies number estimator. We conclude with our main theorem about tolerant testing of $K_k$-freeness (Theorem~\ref{thm:Kk}). 
\begin{algorithm}\caption{\approxKk}\label{alg:approxKk}
Input: $\theta, s, \gamma\in (0,1/2]$
\begin{enumerate}
\item Set $t = \Theta(n{\theta \choose k-1}/(\gamma^2 s))$
\item For $j = 1, \ldots, t$
\begin{enumerate}
    \item Sample a vertex $v\in V$ u.a.r. and a set of $k-1$ indices, $i_1, \ldots, i_{k-1}$ u.a.r. from ${[\theta] \choose k-1}.$
    \item If $v$ is $\theta$-heavy then set $Y_j = 0$ and go to the next iteration. 
    \item Otherwise, query $(v, i_\ell)$ for each $\ell \in [k-1]$. Let $u_1, \ldots, u_{k-1}$ denote the respective return values.
    \item Let $y$ denote the number of $\theta$-light vertices in $\{u_i\}_{i\in [k-1]}$.
    \item If $y < k-2$ then set $Y_j = 0$ (there is more than a single $\theta$-heavy vertex in the set) and go to the next iteration.
    \item If the subgraph induced on $\{u_i\}_{i\in [k-1]}$ is a $(k-1)$-clique then set $Y_j = 1$ w.p. $1/(y+1)$. Otherwise set $Y_j=0$. \label{step:induced}
\end{enumerate}
\item Return $n \cdot {\theta \choose k-1}\cdot \sum_{j\in [t]} Y_j/t$.
\end{enumerate}
\end{algorithm}

\begin{claim}\label{clm:iter}
    For every iteration $j\in [t]$ of Algorithm~\ref{alg:approxKk} it holds that: 
    \begin{equation}
        \mathbb{E}(Y_j) = n_{K_k}(G_\theta)/\left(n \cdot {\theta \choose k-1}\right).
    \end{equation}
\end{claim}

\begin{proof}
    Consider a copy, $\mathcal{K}$, of $K_k$ in $G_\theta$. Let $y(\mathcal{K})$ denote the number of $\theta$-light vertices in $\mathcal{K}$. 
    The probability that $\mathcal{K}$ is found in Step~\ref{step:induced} of \approxKk\ is exactly $y(\mathcal{K}) \cdot 1/(n {\theta \choose k-1})$. Conditioned on the event that indeed $\mathcal{K}$ was found in Step~\ref{step:induced} of iteration $j$ the probability that $Y_j$ is set to be $1$ is $1/y(\mathcal{K})$ (the $\theta$-light vertices of $\mathcal{K}$ are $v$ together with the $y$ light vertices among $u_1,\ldots,u_{k-1}$, so $y+1 = y(\mathcal{K})$). Thus, for every $j$, $\Pr(Y_j = 1) = n_{K_k}(G_\theta)/(n {\theta \choose k-1})$. The claim follows. 
\end{proof}

\begin{claim}\label{clm:nkk}
\approxKk\ is a $K_k$-copies number estimator of the graph $G_{\theta}$. Its query complexity is $O(k^2n{\theta \choose k-1}/(\gamma^2 s))$. 
\end{claim}

\begin{proof}
    For $s \leq n_{K_k}(G_{\theta})$ the claim follows from Claim~\ref{clm:iter}, the setting of $t$ and multiplicative Chernoff's bound (see Theorem~\ref{thm:chermul}). To see the correctness of the second part, consider $s = n_{K_k}(G_{\theta})$. By the first part of the claim, w.h.c.p. the return value of the algorithm is in $(1\pm \gamma) n_{K_k}(G_{\theta}) = (1\pm \gamma) s$. Namely, at most $(1+\gamma) s$. Now consider the case in which $s>n_{K_k}$. By a coupling argument, w.h.c.p. the return value is at most $(1+\gamma) s$ in this case. To see this, couple the return value of the algorithm to the return value of the algorithm on the graph after we add $s-n_{K_k}$ copies to the graph (arbitrarily). Since we added copies, the return value of the algorithm on the modified graph dominates the return value on the original graph (to see this consider executing the algorithm in parallel in both graphs - the outcome in the modified graph is always at least as high) but is guaranteed to be at most $(1+\gamma)s$ w.h.c.p. by the above. The claim about the query complexity follows from construction. This concludes the proof.   
\end{proof}

\begin{algorithm}\caption{\randomKk}\label{alg:randomKk}
    Input: $\theta$, $s$, $\delta$
\begin{enumerate}
\item Repeat $\Theta(n{\theta \choose k-1}\log(1/\delta)/s)$ times: 
\begin{enumerate}    
\item Sample a vertex $v\in V$ u.a.r. and a set of $k-1$ indices, $i_1, \ldots, i_{k-1}$ u.a.r. from ${[\theta] \choose k-1}.$
    \item If $v$ is not $\theta$-light then go to the next iteration.
    \item Otherwise, query $(v, i_\ell)$ for each $\ell \in [k-1]$. Let $u_1, \ldots, u_{k-1}$ denote the respective return values.
    \item Let $y$ denote the number of $\theta$-light vertices in $\{u_i\}_{i\in [k-1]}$.
    \item If $y < k-2$ or if the subgraph induced on $\{u_i\}_{i\in [k-1]}$ is not a clique, then go to the next iteration.
    \item Otherwise, w.p. $1/(y+1)$, return the $k$-clique induced on $v$ and $\{u_i\}_{i\in [k-1]}$.\label{step:it}
\end{enumerate}

\end{enumerate}
\end{algorithm}

\begin{claim}
    The query complexity of \randomKk is $\Theta(k^2 n{\theta \choose k-1}\log(1/\delta)/s)$. If $s< n_{K_k}(G_\theta)$ then it returns a copy w.p. at least $1-\delta$.
\end{claim}
\begin{proof}
    For any copy of $K_k$, $\mathcal{K}$, in each iteration of the loop, the probability that we return $\mathcal{K}$ in Step~\ref{step:it} of the algorithm is exactly $1/(n{\theta \choose k-1})$. Therefore if the number of copies is greater than $s$ then the probability that it will return a copy in a single iteration is at least $s/(n{\theta \choose k-1})$. The claim follows by the setting of the number of iterations.
\end{proof}

\begin{algorithm}\caption{\allNKk}\label{alg:allNKk}
Input: $\theta$, $\mathcal{K}$  
\begin{enumerate}
\item Set $C = \emptyset$
\item If $k=3$, 
reveal the neighbors of all the $\theta$-light vertices of $\mathcal{K}$. Let $S$ denote the set of revealed neighbors. If there is a $\theta$-heavy vertex, $u$, in $\mathcal{K}$, perform a pair query between $u$ and each one of the vertices in $S$.
\item Otherwise, for each $\theta$-light vertex, $v$, of $\mathcal{K}$:
            Reveal all the neighbors of $v$ in $G$ and then for each $\theta$-light neighbor, reveal all its neighbors.
        \item Add to $C$ all the $K_k$-copies that were revealed in the previous steps that belong to $G_\theta$ (namely, that have at most one $\theta$-heavy vertex) and share an edge with $\mathcal{K}$.
        \item Return $C$.
    \end{enumerate}
\end{algorithm}

\begin{claim}\label{clm:query}
    The query complexity of \allNKk\ is $O(k\theta)$ for $k=3$ and $O(k\theta^2)$ for $k>3$.  
\end{claim}

\begin{proof}
    The claim follows from construction.
\end{proof}

\begin{theorem}\label{thm:Kk}
    There exists a tolerant tester for the property of being $K_k$-free that distinguishes graphs that are $\epsilon$-close to being $K_k$-free from graphs that are $\epsilon {k\choose 2}(1+\eta)$-far from being $K_k$-free, where $\epsilon\in (0,1]$ and $\eta\in (0,1]$ are parameters. 
    The expected query complexity of the algorithm is $\tilde{O}\left(\frac{k^4\theta^{2k-3}}{\eta^2\epsilon \bar{d}} + k^5\theta^{2k-4}\theta^{\min\{2, k-2\}}\right)$ where $\theta = \Theta(\arb(G)/(\eta\epsilon))$.
\end{theorem}

\begin{proof}
    The theorem follows from the fact that $\Delta((G_\theta)^H) \leq {k\choose 2} \theta^{k-2}$, Theorem~\ref{thm:general} and Claims~\ref{clm:exp}-\ref{clm:query}.
    \end{proof}

\section{The oracles for $2$-connected motifs of radius $1$}\label{sec:2conn}
In this section, $H$ is a motif which is $2$-connected and has radius $1$.
We make the following observation on the copies of $H$ in $G_\theta$.
\begin{observation}\label{observation1}
For any copy of $H$, $\mathcal{H}$, in $G_\theta$, at least one of the following holds:
\begin{enumerate}
    \item there exists a $\theta$-light vertex in $\mathcal{H}$ which is incident to all the vertices in $\mathcal{H}$.\label{obs1.item1}
    \item there is exactly one $\theta$-heavy vertex in $\mathcal{H}$.
\end{enumerate}
\end{observation}
\begin{proof}
Let $v$ be a vertex in $\mathcal{H}$ which is incident to all other vertices in $\mathcal{H}$ (such vertex exists since $H$ has radius $1$). If $v$ is $\theta$-light then Item~\ref{obs1.item1} holds. Otherwise, $v$ is $\theta$-heavy. Since there are no edges in $G_\theta$ such that both endpoints are $\theta$-heavy, the claim follows. 
\end{proof}
We say that a BFS is {\em $\theta$-restricted} if it explores only the neighbors of the $\theta$-light vertices it reaches. We say that an exploration {\em reveals a copy $\mathcal{H}$} if all the vertices of $\mathcal{H}$ appeared in the exploration. %Therefore, if an exploration reveals a copy $\mathcal{H}$ then $\mathcal{H}$ is a subgraph of the graph induced on the explored vertices.    
\begin{claim}\label{clm:BFS}
    Let $H$ be a motif that is $2$-connected and has radius $1$.
In every copy of $H$ in $G_\theta$, $\mathcal{H}$, for every $\theta$-light vertex $v$ in $\mathcal{H}$, a $\theta$-restricted BFS from $v$ of depth $|V(H)|$ reveals $\mathcal{H}$.
\end{claim}
\begin{proof}
    If item~\ref{obs1.item1} of  observation~\ref{observation1} holds, then clearly the claim follows. Otherwise, there is exactly one $\theta$-heavy vertex, $u$, in $\mathcal{H}$. 
    Since $H$ is $2$-connected, after we remove $u$ from $\mathcal{H}$, $\mathcal{H}$ remains connected and its diameter is at most $|V(H)|-1$. Thus, a $\theta$-restricted BFS of depth $|V(H)|-1$ from any $\theta$-light vertex of $\mathcal{H}$ reveals all the $\theta$-light vertices of  $\mathcal{H}$.
    The claim follows from the fact that at least one of the $\theta$-light vertices in $\mathcal{H}$ is a neighbor of $u$. 
\end{proof}
 Therefore, by Claim~\ref{clm:BFS}, finding for a $\theta$-light vertex, $v$, all the $H$-copies in $G_\theta$ that include $v$, can be done by using $O(\theta^{|V(H)|})$ queries. Consequently, a query for all-neighbors in $(G_\theta)^H$ can be answered by making $O(|V(H)|\theta^{|V(H)|})$ queries to $G$. This can be achieved by going over all vertices of the queried copy, $\mathcal{H}$, and for each $\theta$-light vertex, finding all its copies. From the union of these sets we then remove all the copies that do not share an edge with $\mathcal{H}$ and obtain all the neighbors of $\mathcal{H}$ in $(G_\theta)^H$.

To support random copy queries we assign each copy of $H$ to its $\theta$-light vertex of lowest id. To sample an $H$-copy u.a.r. from $G_\theta$ we first pick a vertex $v\in V$ u.a.r., if it is $\theta$-light then we find all its $H$-copies as described above. We then pick a copy that is assigned to $v$ u.a.r. and return it w.p. $a_{H, \theta}(v)/a_{H,\theta}^{max}$ where $a_{H, \theta}(v)$ denotes the number of copies that are assigned to $v$ in $G_\theta$ and $a_{H, \theta}^{max}$ denotes the maximum number of copies that include a specific $\theta$-light vertex (recall that the copies in $G_\theta$ are assigned only to $\theta$-light vertices). Clearly, the above algorithm samples each copy of $H$ in $G_\theta$ w.p. $1/(n\cdot a_{H, \theta}^{max})$. Consequently, the probability that it returns a copy is $n_H(G_\theta)/(n\cdot a_{H,\theta}^{max})$. 
Thus, the success probability of the algorithm depends on $a_{H,\theta}^{max}$.
A straightforward upper-bound on $a_{H,\theta}^{max}$ is $|V(H)|!{\theta^{|V(H)|} \choose |V(H)|}$ where the first term is due to all possible labeling of the vertices and the second term is due to all possible ways to pick the vertices of the copy from the set of vertices we explored.  
In the next claim we provide a tighter bound on $a_{H,\theta}^{max}$.

\begin{claim}\label{clm:motif}
     For any motif $H$, it holds that $a_{H,\theta}^{max} \leq |V(H)|\cdot\theta^{|V(H)|-1}$.
\end{claim}
\begin{proof}
Fix an arbitrary labeling of $H$. Let  $v\in V$. We next bound $a_{H, \theta}(v)$.
Label the $H$-copies that $v$ belongs to according to the fixed labeling of $H$ (each copy has its own labeling). If there is more than one such labeling per copy (due to automorphism), pick one arbitrarily. We next describe a one-to-one mapping from the set of labeled copies to a set of sequences of integers. We then bound the cardinality of the latter set and obtain a bound on the number of $H$-copies that $v$ belongs to and hence a bound on $a_{H,\theta}^{max}$.
Fix a labeled copy $\mathcal{H}$ containing $v$, and consider the BFS exploration of $\mathcal{H}$ starting at $v$, where vertices with smaller label are explored first. Given the label of $v$ and the sequence of the vertices in $\mathcal{H}$ by their BFS order we can reconstruct the labels of all the vertices in the copy. The set of vertices and their labels are mapped to a single copy.  
Moreover, by the above, we can identify the copy by the identity of the root, its label and a sequence of $|V(H)|-1$ indices in $[\theta]$. Each index indicates the index of the respective element in its parent adjacency-list. The indices are listed according to the BFS exploration. Since the identity of the root is known, we can inductively reconstruct the tree by the sequence of indices. Note that the parents are always $\theta$-light (while the leaves may be $\theta$-heavy) therefore it suffices to use indices in $[\theta]$.
From the above we see that for a specific label of $v$ there are at most $\theta^{|V(H)|-1}$ copies of $H$ to which $v$ belongs. Since there are $|V(H)|$ possible labels, we get the desired bound.  
\end{proof}
\begin{algorithm}\caption{\approxH}\label{alg:approxH}
Input: $s, \gamma\in (0,1/2], \theta$
\begin{enumerate}
\item Set $t = \Theta(n\cdot a_{H,\theta}^{max}/(\gamma^2s))$
\item For $j = 1, \ldots, t$
\begin{enumerate}
    \item Sample a vertex $v\in V$ u.a.r. 
    \item If $v$ is not $\theta$-light then set $Y_j = 0$ and go to the next iteration. 
    \item Otherwise, perform a $\theta$-restricted BFS from $v$ to depth $|V(H)|$ and reveal all the $H$-copies assigned to $v$ in $G_\theta$. Set $Y_j$ to be this number.
    \end{enumerate}
\item Return $n \cdot \sum_{j\in [t]} Y_j/t$.
\end{enumerate}
\end{algorithm}

\begin{claim}\label{clm:nH}
\approxH\ is an $H$-copies number estimator of the graph $G_{\theta}$. Its query complexity is $O(n\cdot a_{H,\theta}^{max}\cdot \theta^{|V(H)|}/(\gamma^2s))$. 
\end{claim}

\begin{proof}
The claim about the query complexity follows by construction. 
Consider the random variable $Y_j$ which is defined in \approxH. 
Since $\mathbb{E}(Y_j) \geq n_H(G_\theta)/n$ and $Y_j \leq a_{H,\theta}^{max}$, for $s\leq n_H(G_\theta)$, 
by the setting of $t$ and the multiplicative Chernoff's bound (Theorem~\ref{thm:chermul}), w.h.c.p. the output of the algorithm is in $(1\pm \gamma)n_H(G_\theta)$, as required. For $s> n_H(G_\theta)$, by a coupling argument, w.h.c.p. the output of the algorithm is at most $(1+\gamma)s$ (see more details in the proof of Claim~\ref{clm:nkk}).
\end{proof}

\begin{algorithm}\caption{\randomH}\label{alg:randomH}
    Input: $\theta$, $s$, $\delta$
\begin{enumerate}
\item Repeat $\Theta(n\cdot a_{H,\theta}^{max}\log(1/\delta)/s)$ times: 
\begin{enumerate}    
\item Sample a vertex $v\in V$ u.a.r.
    \item If $v$ is $\theta$-heavy then go to the next iteration.
    \item Otherwise, perform a $\theta$-restricted BFS from $v$ to depth $|V(H)|$ and reveal all the $H$-copies assigned to $v$ in $G_\theta$.
    \item Pick a random copy from the $H$-copies assigned to $v$ and return it with probability $a_{H,\theta}(v)/a_{H,\theta}^{max}$.
\end{enumerate}

\end{enumerate}
\end{algorithm}

\begin{claim}
    The query complexity of \randomH is $\Theta(\theta^{|V(H)|} \cdot n\cdot a_{H,\theta}^{max}\log(1/\delta)/s)$. It returns an $H$-copy w.p. at least $1-\delta$.
\end{claim}
\begin{proof}
    The claim about the query complexity follows by construction. In each iteration of the algorithm, each copy of $H$ is returned with probability $1/(n\cdot a_{H,\theta}^{max})$. Thus, in each iteration a copy is returned w.p. $n_H(G_\theta)/(n\cdot a_{H,\theta}^{max})$. Thus, if $s \leq n_H(G_\theta)$, a copy is returned w.p. at least $1-\delta$ by the setting of the number of iterations. 
\end{proof}
\begin{algorithm}\caption{\allNH}\label{alg:allnH}
Input: $\theta$, $\mathcal{H}$
\begin{enumerate}
\item Set $C = \emptyset$
        \item For each $\theta$-light vertex $v$ of $\mathcal{H}$:
        \begin{enumerate}
            \item Reveal all the $H$-copies of $v$ in $G_\theta$ by performing a $\theta$-restricted BFS to depth $|V(H)|$ from $v$ in $G$.
            \item Add to $C$ all the revealed $H$-copies that share an edge with $\mathcal{H}$ (every such copy contains a $\theta$-light vertex of $\mathcal{H}$, since $G_\theta$ has no edge between two $\theta$-heavy vertices).
        \end{enumerate}
        \item Return $C$.
    \end{enumerate}
\end{algorithm}

\begin{claim}\label{clm:queryH}
    The query complexity of \allNH\ is $O(|V(H)|\theta^{|V(H)|})$.  
\end{claim}

\begin{proof}
    By construction.
\end{proof}

\begin{theorem}\label{thm:H}
    There exists a tolerant tester for the property of being $H$-free for $H$ that is $2$-connected and has radius $1$ that distinguishes graphs that are $\epsilon$-close to being $H$-free from graphs that are $\epsilon |E(H)|(1+\eta)$-far from being $H$-free, where $\epsilon\in (0,1]$ and $\eta\in (0,1]$ are parameters. 
    The expected query complexity of the algorithm is $\tilde{O}\left(\left(h^7/\eta^2+ h^3/(\eta \epsilon \bar{d})\right)\cdot \theta^{3h-2}\right)$ where $\theta = \Theta(\arb(G)/(\eta\epsilon))$ and $h = |V(H)|$.
\end{theorem}

    \begin{proof}
    The theorem follows from the fact that $\Delta((G_\theta)^H) \leq h \cdot a_{H,\theta}^{max}$, Theorem~\ref{thm:general} and Claims~\ref{clm:motif}-\ref{clm:queryH}.
    \end{proof}

\section{The tester for the bounded-degree model}
Clearly, \tester\ applies also to bounded degree graphs. Moreover, for the setting where $\theta = d+1$ the graph $G_{\theta}$ is simply $G$ and all the vertices in the graph are  $\theta$-light. Consequently, we can test $H$-freeness for {\em any} motif $H$ by using the oracles in Section~\ref{sec:2conn} without making any modifications\footnote{In fact the query complexity of Algorithms~\ref{alg:approxH}--\ref{alg:allnH} can be slightly improved by setting the depth of the BFS exploration to be the diameter of $H$ (instead of $|V(H)|$).}. We obtain the following result. 

\begin{theorem}\label{thm:bounded}
    There exists a tolerant tester, in the bounded degree model, for the property of being $H$-free for any motif $H$. Specifically the tester distinguishes graphs that are $\epsilon$-close to being $H$-free from graphs that are $\epsilon |E(H)|(1+\eta)$-far from being $H$-free, where $\epsilon\in (0,1]$ and $\eta\in (0,1]$ are parameters. 
    The expected query complexity of the algorithm is $\tilde{O}\left(\left(h^7/\eta^2+ h^3/(\eta \epsilon \bar{d})\right)\cdot d^{3h-2}\right)$ where $h = |V(H)|$ and $d$ is the degree bound.
\end{theorem}

\begin{proof}
    By the proof of Theorem~\ref{thm:H}.
\end{proof}

\appendix
\section{Concentration bound}
\begin{theorem}[Multiplicative Chernoff's Bound]\label{thm:chermul}
Let $X_1, \ldots, X_n$ be independent and identically distributed random variables ranging in $[0, 1]$, and let $p = \mathbb{E}[X_1]$. Then, for every $\gamma \in (0, 2]$, it holds that
\begin{equation}
\Pr\left[ \left| \frac{1}{n} \cdot \sum_{i\in [n]} X_i - p \right| > \gamma \cdot p\right] < 2 \cdot e^{-\gamma^2 pn /4} \;.
\end{equation}  
\end{theorem}

\bibliography{lit}

\end{document}